\documentclass[a4paper,english,10pt]{article}

\usepackage[utf8]{inputenc}
\usepackage[a4paper]{geometry}
\usepackage{cite}
\usepackage{enumerate,enumitem,hyperref,nameref}
\usepackage{verbatim}
\usepackage{amsmath,amsthm,amsfonts,amssymb,stmaryrd,mathtools,esvect}
\usepackage{placeins}
\usepackage{tikz}
\usetikzlibrary{calc,fit,shapes,positioning,arrows,intersections}

\theoremstyle{plain}\newtheorem{theorem}{Theorem}
\theoremstyle{plain}
\theoremstyle{plain}\newtheorem{proposition}[theorem]{Proposition}
\theoremstyle{plain}
\theoremstyle{plain}\newtheorem{definition}{Definition}
\theoremstyle{plain}
\theoremstyle{plain}

\makeatletter
\let\orgdescriptionlabel\descriptionlabel
\renewcommand*{\descriptionlabel}[1]{%
  \let\orglabel\label
  \let\label\@gobble
  \phantomsection
  \edef\@currentlabel{\ignorespaces #1\unskip}%
  \let\label\orglabel
  \orgdescriptionlabel{#1}%
}
\makeatother

\newcommand{\defeq}{\triangleq}
\newcommand{\defiff}{\stackrel{\triangle}{\iff}}

\newcommand{\mono}{\rightarrowtail}

\newcommand{\emb}{\hookrightarrow}
\newcommand{\face}[1]{\langle #1 \rangle}

\newcommand{\rng}{\textsf{\textit{img}}}
\newcommand{\prnt}{\textsf{\textit{prnt}}}
\newcommand{\ctrl}{\textsf{\textit{ctrl}}}
\newcommand{\link}{\textsf{\textit{link}}}
\newcommand{\ephi}[1]{\phi^\mathsf{#1}}
\newcommand\restr[2]{{\left.\kern-\nulldelimiterspace#1\vphantom{\big|}\right|_{#2}}}

\hypersetup{
	colorlinks=true,
	linkcolor=black,
	citecolor=black,
	filecolor=black,
	urlcolor=black,
	pdftitle={A CSP implementation of the bigraph embedding problem},
	pdfauthor={Marino Miculan, Marco  Peressotti},
	pdfsubject={Concurrency Theory},
	pdfkeywords={Bigraphical reactive systems, Constraint programming}
}

\title{A CSP implementation of the bigraph embedding problem}
\author{
	\begin{tabular}{ccc}
	Marino Miculan&\qquad& Marco  Peressotti\\
	\small\href{mailto:marino.miculan@uniud.it}{\tt marino.miculan@uniud.it}
	&\qquad&
	\small\href{mailto:marco.peressotti@uniud.it}{\tt marco.peressotti@uniud.it}
	\end{tabular}\\
	\small	Laboratory of Models and Applications of Distributed Systems \\[-.8ex]
	\small	Department of Mathematics and Computer Science\\[-.8ex]
	\small	University of Udine, Italy\\
}
\date{}

\begin{document}

\maketitle

\begin{abstract}
  A crucial problem for many results and tools about bigraphs and bigraphical reactive systems is bigraph embedding.
  An embedding is more informative than a bigraph matching, since it keeps track of the correspondence between the various components of the redex (guest) within the agent (host). 
  In this paper, we present an algorithm for computing embeddings based on a reduction to a
  \emph{constraint satisfaction} problem. This algorithm, that we prove to be
  sound and complete, has been successfully implemented in LibBig, a library for manipulating bigraphical reactive systems. This library can be used for implementing a wide range of tools,
  and it can be adapted to various extensions of bigraphs.
\end{abstract}

\section{Introduction}
\emph{Bigraphical Reactive Systems} (BRSs)
\cite{jm:popl03,milner:bigraphbook} are a flexible and
expressive meta-model for ubiquitous computation.
System states are represented by \emph{bigraphs}, which are
compositional data structures describing at once both the locations
and the logical connections of (possibly nested) components of a
system.
Like graph rewriting \cite{graphtransformation}, the dynamic
behaviour of a system is defined by a set of \emph{(parametric)
  reaction rules}, which can modify a bigraph by replacing a
\emph{redex} with a \emph{reactum}, possibly changing agents'
positions and connections.

BRSs have been successfully applied to the formalization of a broad
variety of domain-specific calculi and models, from traditional
programming languages to process calculi for concurrency and mobility,
from context-aware systems to web-service orchestration languages,
from business processes to systems biology; a non exhaustive list is
\cite{bdehn:fossacs06,bghhn:coord08,bgm:biobig,dhk:fcm,mmp:dais14,mp:br-tr13}.
Very recently bigraphs have been used in structure-aware agent-based
computing for modelling the structure of the (physical) world where the
agents operates (e.g., drones, robots, etc.) \cite{pkss:bigactors}.

Beside their normative and expressive power, BRSs are appealing
because they provide a range of interesting general results and tools,
which can be readily instantiated with the specific model under
scrutiny: simulation tools, systematic construction of compositional
bisimulations \cite{jm:popl03}, graphical editors \cite{fph:gcm12},
general model checkers \cite{pdh:sac12}, modular composition
\cite{pdh:refine11}, stochastic extensions \cite{kmt:mfps08}, etc.

\looseness=-1
In this paper, we give an implementation for a crucial problem that
virtually all these tools have to deal with, i.e., the matching a
bigraph inside an agent. Roughly, this can be stated as follows: given
$R$ and $A$, we have to find (all, or some) $C, D$ such that $A =
C\circ R \circ D$.  Clearly this is required by any simulation tool
(in order to apply a reaction rule, we have to match the redex inside
the agent, and then replace it with the reactum), but also in other
tools, e.g., for implementing ``find\&replace'' in graphical editors,
for occurrence checks in sortings \cite{bg:calco09} and model
checkers, for refinements in architectural design tools, etc.

Like the similar and well-known \emph{subgraph isomorphism} problem,
bigraph matching is NP-complete (see \cite{bmr:tgc14}). 
However, using the theory of
Fixed Parameter Tractability it can be shown that the exponential
explosion depends only on the size (more precisely, the width) of the
redex to be found, and not on the size of the agent. In most practical
cases, this width is constant and small (e.g.~$\leq 3$), hence the
problem becomes feasible.

A (rather \emph{ad-hoc}) implementation of bigraph matching has been given
in the \emph{BPLTool} \cite{gdbm13:indmatch}; this was based on a
term-based representation of agents and rules, in the spirit of term
rewriting systems.  More recently, a more graphical-oriented approach
has been preferred.  H{\o}jsgaard have introduced the notion of
\emph{bigraph embedding} \cite{hoesgaard:thesis}, which is a function
from nodes and edges of the redex to nodes and edges of the agent,
describing how the former is embedded in the latter.  Although
embeddings and matchings are basically equivalent, embeddings turn out
to be more useful especially in connection with Gillespie-like
algorithms for stochastic simulations \cite{hoesgaard:thesis}, because
they allow for a simpler calculation of interference between redexes.
In fact, embeddings are at the core of the \emph{Bigraphic Abstract
  Machine} \cite{perrone:thesis}, a general abstract machine for
implementing various kinds of BRSs, with several possible execution
strategies.

For these reasons, in this paper we focus on the \emph{bigraph
  embedding problem}.  More precisely, we translate the embedding
problem to a \emph{constraint satisfaction problem} (CSP), whose
solutions correspond to bigraph embeddings. Instead of defining
directly a CSP for the bigraph embedding problem we take advantage of
bigraphs being the ``merge'' of two graphical structures (called 
\emph{link graphs} and \emph{place graphs} respectively): initially we
define the encoding for embeddings of these two orthogonal structures
separately and then combine them by means of some consistency constraints
reflecting the interplay between link and place structures. This split
mimics the peculiar structure of bigraphs and allows us to factor the exposition
of the problem, its encoding and the accompanying adequacy results.
An implementation based on the
\emph{CHOCO} %\footnote{\url{http://www.emn.fr/z-info/choco-solver/}}
solver is available in \emph{LibBig} (available at \href{https://github.com/bigraphs/jlibbig}{\tt https://github.com/bigraphs/jlibbig}),
an extensible library for manipulating bigraphical reactive systems.
However, this is an implementation choice mainly due to the use of Java, but
the results of this paper can be implemented in any solver capable of
handling the integer solutions of a linear equation system.

We do not provide an exhaustive discussion of experimental results
because the encoding proposed is ``solver-independent'' and moreover,
because there is no widely-acknowledged benchmark suit for this problem.
In fact, finding a reasonably representative set of instances still
is an open question.

\paragraph{Synopsis}
In Section~\ref{sec:brs} we briefly recall the notion of bigraphs and
bigraphical reactive systems and in Section \ref{sec:emb} present the
bigraph embedding problem, its complexity and how it can be divided in
two sub-problems, by taking advantage of the components of
bigraphs. The implementation of the bigraph embedding problem as a 
constraint satisfaction problem and the adequacy results 
are presented in Sections~\ref{sec:csp}.
Conclusions, with some experimental evaluations, and some directions
for future work are in Section~\ref{sec:concl}.

\begin{figure}[t]
  \centering
\begin{tikzpicture}[font=\small]
		\node[label=-90:Bigraph $G:\face{3,Y}\to\face{2,X}$](big) at (0,0){
		\begin{tikzpicture}[remember picture,
			every text node part/.style={align=left},
			dot/.style={circle,fill=black,minimum size=2pt,inner sep=0, outer sep=0},
			region/.style={rectangle, draw=gray,dashed,rounded corners,inner sep=5pt},
			site/.style={rectangle, draw=gray,fill=gray!10,dashed,rounded corners,
				outer sep=2pt,inner sep=5pt,minimum size=20pt},
			lbl/.style={inner sep=0,outer sep=2pt},
			nam/.style={inner sep=0,outer sep=2pt},
			ctrA/.style={ellipse,draw=black,solid,inner sep=5pt,outer sep=0},
			ctrB/.style={rectangle,draw=black,solid,inner sep=5pt,outer sep=0},
			link/.style={draw=black!30!green,thick},
			anc/.style={draw,circle,fill,black!30!green,minimum size=0,
				outer sep=0, inner sep=0}]
			\node[region](r0) at (0,0) {0\\
			    \begin{tikzpicture}
					\node[ctrB] (c0) at (0,0){
						\begin{tikzpicture}
			    			\node[site] (s1) at (0,0) {1};
							\node[ctrA, right=5pt of s1] (c2){};
							\node[lbl, right=0 of c2] {\(s\)};
							\node[dot, above right=0 of c2] (p2-0) {};
							\node[dot, below left=0 of c2] (p2-1) {};
			    		\end{tikzpicture}};
					\node[dot, below=0 of c0] (p0-0) {};
					\node[dot, left=0 of c0,yshift=-2pt] (p0-1) {};
					\node[lbl,above=0 of c0] {\(r\)};
			    \end{tikzpicture}
			  };
			\node[region, right=5pt of r0](r1) {1\\
			    \begin{tikzpicture}
			    	\node[site] (s0) at (0,0) {0};
					\node[ctrB, right=5pt of s0] (c1) {
					\begin{tikzpicture}
			    		\node[site] (s2) at (0,0) {2};
			    	\end{tikzpicture}};
					\node[dot, above=0 of c1, xshift=-4pt] (p1-0) {};
					\node[dot, above=0 of c1, xshift=4pt] (p1-1) {};
					\node[lbl,below=0 of c1] {\(t\)};
			    \end{tikzpicture}
			  };
			\node[nam,above=5pt of r0] (n0) {\(x_0\)};
			\node[nam,above=3pt of r1] (n1) {\(x_1\)};
			\node[nam,below=10pt of p0-0] (n2) {\(y_0\)};
			\node[anc, above=16pt of c2, xshift=13pt] (e0) {};
			\draw[link,bend left] (e0)to node {}(n0);
			\draw[link,bend left] (e0)to node {}(p2-0);
			\draw[link,bend left] (e0)to node {}(p1-0);
			\draw[link,bend right=50] (p0-1) to node {}(p2-1);
			\draw[link,bend right] (p1-1)to node {}(n1);
			\draw[link] (p0-0)--(n2);
		\end{tikzpicture}};

		\node[label=90:Place graph $G^P:3\to 2$, 
			left=10pt of big]{
		\begin{tikzpicture}[scale=.65,
			region/.style={rectangle, draw=gray,dashed,rounded corners,
				outer sep=1pt,inner sep=4pt},
			site/.style={rectangle, draw=gray,fill=gray!10,dashed,
				rounded corners,outer sep=1pt,inner sep=4pt},]
			\node[region] (r0) at (-1.5,3) {0};
			\node[region] (r1) at (.5,3) {1};
			\node[site] (s0) at (0,0) {0};
			\node[site] (s1) at (-2,0) {1};
			\node[site] (s2) at (1,0) {2};
			\node[below=10pt of r0] (c0) {r};
			\node[above=10pt of s2] (c1) {t};
			\node (c2) at (-1.1,.7) {s};
			\draw[thick, draw=gray] (r0.south)--(c0.north);
			\draw[thick, draw=gray] (c0)--(s1);
			\draw[thick, draw=gray] (c0)--(c2);
			\draw[thick, draw=gray] (r1)--(c1)--(s2);
			\draw[thick, draw=gray] (r1)--(s0);
		\end{tikzpicture}};
		
		\node[label=90:Link graph $G^L:Y\to X$,
			right=10pt of big](lnk){
		\begin{tikzpicture}[
			every text node part/.style={align=left},
			dot/.style={circle,fill=black,minimum size=2pt,inner sep=0, outer sep=0},
			lbl/.style={inner sep=0,outer sep=2pt},
			nam/.style={inner sep=0,outer sep=2pt},
			ctrA/.style={ellipse,draw=black,solid,inner sep=5pt,outer sep=0},
			ctrB/.style={rectangle,rounded corners,draw=black,solid,minimum size=15pt,inner sep=5pt,outer sep=0},
			link/.style={draw=black!30!green,thick},
			anc/.style={draw,circle,fill,draw=black!30!green,minimum size=0,outer sep=0, inner sep=0}]
			\node[ctrB] (c0) at (0,0) {\(r\)};
			\node[dot, below=0 of c0] (p0-0) {};
			\node[dot, left=0 of c0] (p0-1) {};
			%\node[lbl,right=0 of c0] (l0) {\(r\)};
			\node[ctrA, above left=8pt of c0,inner sep=3pt] (c2){\(s\)};
			%\node[lbl, left=0 of c2] (l2) {\(s\)};
			\node[dot, above=0 of c2] (p2-0) {};
			\node[dot, below=0 of c2] (p2-1) {};
			\node[ctrB, right=16pt of c2] (c1) {\(t\)};
			\node[dot, above=0 of c1, xshift=-3pt] (p1-0) {};
			\node[dot, above=0 of c1, xshift=3pt] (p1-1) {};
			%\node[lbl,right=0 of c1] (l1) {\(t\)};
			\node[nam,above=12pt of c2] (n0) {\(x_0\)};
			\node[nam,above=10pt of c1] (n1) {\(x_1\)};
			\node[nam,below=10pt of c0] (n2) {\(y_0\)};
			\node[anc, above=3pt of c2, xshift=12pt] (e0) {};
			\draw[link,bend left] (e0)to node {}(n0);
			\draw[link,bend left] (e0)to node {}(p2-0);
			\draw[link,bend left] (e0)to node {}(p1-0);
			\draw[link,bend left] (p0-1) to node {}(p2-1);
			\draw[link,bend right] (p1-1)to node {}(n1.south);
			\draw[link] (p0-0)--(n2);
		\end{tikzpicture}};
\end{tikzpicture}
	\caption{Forming a bigraph from a place graph and a link graph.}
	\label{fig:bigraph-comp}
\end{figure}

\section{Bigraphical reactive systems}
\label{sec:brs}

In this section we briefly recall the notion of Bigraphical Reactive
Systems (BRS) referring the interested reader to
\cite{milner:bigraphbook}.

The key point of BRSs is that ``the model should consist in some sort
of reconfigurable space''. Agents may interact in this space, even if
they are spatially separated.  This means that two agents may be
adjacent in two ways: they may be at the same \emph{place}, or they
may be connected by a \emph{link}. This leads to the definition of
\emph{bigraphs} as a data structure for representing the state of the
system. A bigraph can be seen as enriched hyper-graph combining two
independent graphical structures over the same set of \emph{nodes}: a
hierarchy of \emph{places}, and a hyper-graph of \emph{links}.

\begin{definition}[Bigraph {\cite[Def.~2.3]{milner:bigraphbook}}]
A \emph{bigraph graph} $G$ over a given signature $\Sigma$ 
(i.e.~a set of types, called controls, denoting a finite arity)
is a an object 
$$(V_G, E_G, \ctrl_G, \prnt_G, \link_G):\face{n_G,X_G}\to\face{m_G,Y_G}$$
composed by two substructures (Figure~\ref{fig:bigraph-comp}): 
a \emph{place graph}  $G^P=(V_G, \ctrl_G, \prnt_G):{n_G\to m_G}$ and a
\emph{link graph} $G^L=(V_G,E_G,\ctrl_G,\link_G):{X_G\to Y_G}$.
The set $V_G$ is a finite set of nodes and to each of them is assigned a
control in $\Sigma$ by the \emph{control map} $\ctrl_G : V_G\to \Sigma$.
The set $E_G$ is a finite set of names called \emph{edges}.

These structures presents an inner interface (composed by $n_G$ and 
$X_G$) and an outer one ($m_G$, $Y_G$) along which can be composed with
other of their kind as long as they do not share any
node or edge. In particular, $X_G$ and $Y_G$ are finite sets
of names and $n_G$ and $m_G$ are finite ordinals that index
\emph{sites} and \emph{roots} respectively.

On the side of $G^P$, nodes, sites and roots are organized in
a forest described by the \emph{parent map} $\prnt_G : V_G 
\uplus n_G \to V_G \uplus m_G$ such that sites are leaves and 
roots are exactly $m_G$.

On the side of $G^L$, nodes, edges and names of the inner and outer
interface forms a hyper-graph described by the \emph{link map} 
$\link_G : P_G\uplus X_G \to E_G \uplus Y_G$ which is a function from $X_G$ and ports $P_G$ (i.e.~elements of 
the finite ordinal associated to each node by its control) to 
edges $E_G$ and names in $Y_G$.
\end{definition}

\looseness=-1
The dynamic behaviour of a system is described in terms of
\emph{reactions} of the form $a \rightarrowtriangle a'$ where $a,a'$
are agents, i.e.~bigraphs with inner interface $\face{0,\emptyset}$.
Reactions are defined by means of graph rewrite rules, which are pairs
of bigraphs $(R_L, R_R)$ equipped with a function $\eta$ from the 
sites of $R_R$ to those of $R_L$ called \emph{instantiation rule}.
A bigraphical encoding for the open reaction rule of the Ambient 
Calculus is shown in Figure~\ref{fig:amb-open} where redex and 
reactum are the bigraph on the left and the one on the right respectively
and the instantiation rule is drawn in red. A rule fires when its redex 
can be embedded into the agent; then, the matched part is replaced by 
the reactum and the parameters (i.e.~the substructures determined by 
the sites of the redex) are instantiated accordingly with $\eta$.

\begin{figure}[t]
  \centering
  \begin{tikzpicture}[font=\small,remember picture,
      extended/.style={shorten >=-#1, shorten <=-#1},
      extended/.default=8pt,
      extended end/.style={shorten >=-#1},
      extended start/.style={shorten <=-#1},
      every text node part/.style={align=left},
      dot/.style={circle,fill=black,minimum size=2pt,inner sep=0, outer sep=0},
      region/.style={rectangle, draw=gray,dashed,rounded corners,inner sep=5pt},
      site/.style={rectangle, draw=gray,fill=gray!10,dashed,rounded corners,
        outer sep=2pt,inner sep=5pt,minimum size=20pt},
      lbl/.style={inner sep=0,outer sep=2pt},
      nam/.style={inner sep=0,outer sep=2pt},
      ctrAmb/.style={ellipse,draw=black,solid,inner sep=1pt,outer sep=0},
      ctrCap/.style={rectangle,draw=black,solid,inner sep=5pt,outer sep=0},
      link/.style={draw,black!30!green,thick},
      anc/.style={draw,circle,black!30!green,minimum size=0,
        outer sep=0, inner sep=0},
      tap/.style={draw,circle,fill,draw,black!30!green,minimum size=1pt,
        outer sep=0, inner sep=0}]
      \node(redex) at (0,0){
        \begin{tikzpicture}
          \node[region,inner sep=8pt](root) at (0,0) {
            \begin{tikzpicture}
            \node[ctrCap] (c0) at (0,0){
              \begin{tikzpicture}
                  \node[site] (s0) at (0,0) {0};
                \end{tikzpicture}};
            \node[dot,above=0 of c0] (c0-0) {};
            \node[lbl,below=0 of c0] {\(open\)};
            \node[ctrAmb, right=15pt of c0] (c1){
              \begin{tikzpicture}
                  \node[site] (s1) at (0,0) {1};
                \end{tikzpicture}};
            \node[dot,above=0 of c1] (c1-0) {};
            \node[lbl,below=0 of c1] {\(amb\)};
            \end{tikzpicture}
          };
          \node[nam,above=5pt of root] (x) {\(x\)};
          \draw[link,out=-90,in=60] (x) to node {}(c0-0);
          \draw[link,out=-90,in=120] (x) to node {}(c1-0);
        \end{tikzpicture}};
      \node[right=20pt of redex,yshift=3pt] (reactum)  {
        \begin{tikzpicture}
          \node[region, inner sep=10pt](root) at (0,0) {
            \begin{tikzpicture}
              \node[site] at (0,0) (s2){0};
              \node[site,right=5pt of s2] (s3){1};
            \end{tikzpicture}};
          \node[nam,above=5pt of root] (x) {\(x\)};
          \node[tap]  (x-tap) [below=11pt of x,xshift=2pt] {};
          \draw[link] (x.south) to[bend right] (x-tap);
        \end{tikzpicture}};
      \draw[thick, draw=red, extended, ->] (s2) to[out=200,in=-20] (s0);
      \draw[thick, draw=red, extended, ->] (s3) to[out=200,in=-20] (s1);
      \draw[thick, -open triangle 45] (redex)--(reactum);
      \node[lbl] (caption) at ($(redex.south)!.5!(reactum.south)+(0pt,-20pt)$)
        {\(\mathsf{open}_x.\fbox{\scriptsize 0}\parallel
           \mathsf{amb}_x.\fbox{\scriptsize 1}\rightarrowtriangle
           \fbox{\scriptsize 0}\parallel\fbox{\scriptsize 1}\)};
  \end{tikzpicture}
  \caption{Open reaction rule of the Ambient Calculus.}
  \label{fig:amb-open}
\end{figure}

\section{Bigraph embeddings}
\label{sec:emb}

In this Section we briefly recall the notion of \emph{bigraph embedding}.
The following definitions are taken from \cite{hoesgaard:thesis},
with minor modification to simplify the presentation of the 
equivalent CSP formulation (cf.~Section~\ref{sec:csp}).
As usual, we shall exploit the orthogonality of the link and place graphs, 
by defining \emph{link and place graph embeddings} separately and then 
combine them to extend the notion to bigraph.

\paragraph{Link graph}
Intuitively an embedding of link graphs is a structure preserving map
from one link graph (the \emph{guest}) to another (the \emph{host}). 
As one would expect from a graph 
embedding, this map contains a pair of injections: one for the nodes
and one for the edges (i.e., a support translation). The remaining
of the embedding map specifies how names of the inner and outer interfaces should be mapped into the host link graph. Outer names can be mapped to any link; here injectivity is not required since a context can alias outer names. Dually, inner names can mapped to hyper-edges linking sets of points in the host link graph and such that every point is contained in at most one of these sets.

\begin{definition}[Link graph embedding {\cite[Def~7.5.1]{hoesgaard:thesis}}]\label{def:lge}
	Let $G : X_G \to Y_G$ and $H : X_H \to Y_H$ be two concrete link graphs. 
	A \emph{link graph embedding} $\phi : G \emb H$ is a map
	$\phi \defeq \ephi v \uplus \ephi e \uplus \ephi i \uplus \ephi o$
	(assigning nodes, edges, inner and outer names respectively)
	subject to the following conditions:
	\begin{description}
	\item[(LGE-1)\label{def:lge-1}]
		$\ephi v : V_G \mono V_H$ and $\ephi e : E_G \mono E_H$ are injective;
	\item[(LGE-2)\label{def:lge-2}]
		$\ephi i : X_G \mono \wp(X_H \uplus P_H)$ is fully
                injective: $\forall x\neq x' : \ephi i(x) \cap \ephi i (x') = \emptyset$;
	\item[(LGE-3)\label{def:lge-3}]
		$\ephi o : Y_G \to E_H \uplus Y_H$ in an arbitrary partial map;
	\item[(LGE-4)\label{def:lge-4}]
		$\rng(\ephi e) \cap \rng(\ephi o) = \emptyset$ and $\rng(\ephi{port})\cap \bigcup\rng(\ephi i)  = \emptyset$;
	\item[(LGE-5)\label{def:lge-5}]
		$\ephi p \circ \restr{\link_G^{-1}}{E_G} = \link_H^{-1}\circ \ephi e$;
	\item[(LGE-6)\label{def:lge-6}]
		$\ctrl_G = \ctrl_H \circ \ephi v$;
	\item[(LGE-7)\label{def:lge-7}]
		$\forall p \in X_G \uplus P_G : \forall p' \in (\ephi p)(p) : (\ephi h \circ \link_G)(p) = \link_h(p')$
	\end{description}
	where 
	$\ephi p \defeq \ephi i \uplus \ephi{port}$,
	$\ephi h \defeq \ephi e \uplus \ephi{o}$ and  
	$\ephi{port}:P_G \mono P_H$ is $\ephi{port}(v,i) \defeq (\ephi v(v),i))$.
\end{definition}

The first three conditions are on the single sub-maps of the 
embedding. Condition \ref{def:lge-4} ensures that no components 
(except for outer names) are identified; condition \ref{def:lge-5}
imposes that points connected by the image of an edge are all 
covered. Finally, conditions \ref{def:lge-6} and \ref{def:lge-7} 
ensure that the guest structure is preserved i.e.~node controls 
and point linkings are preserved.

\paragraph{Place graph}
Like link graph embeddings, place graph embeddings
are just a structure preserving injective map from nodes along with suitable maps for the inner and outer interfaces. In particular, a site is mapped to the set of sites and nodes that are ``put under it'' and a root is mapped to the host root or node that is ``put over it'' splitting the host place graphs in three parts: the guest image, the context and the parameter (which are above and below the guest image).

\begin{definition}[Place graph embedding {\cite[Def~7.5.4]{hoesgaard:thesis}}]\label{def:pge}
	Let $G : n_G \to m_G$ and $H : n_H \to m_H$ be two concrete place graphs. 
	A \emph{place graph embedding} $\phi : G \emb H$ is a map
	$\phi \defeq \ephi v \uplus \ephi s \uplus \ephi r$
	(assigning nodes, sites and regions respectively)
	subject to the following conditions:
	\begin{description}
	\item[(PGE-1)\label{def:pge-1}]
		$\ephi v : V_G \mono V_H$ is injective;
	\item[(PGE-2)\label{def:pge-2}]
		$\ephi s : n_G \mono \wp(n_H \uplus V_H)$ is fully injective;
	\item[(PGE-3)\label{def:pge-3}]
		$\ephi r : m_G \to V_H \uplus m_H$ in an arbitrary map;
	\item[(PGE-4)\label{def:pge-4}]
		$\rng(\ephi v) \cap \rng(\ephi r) = \emptyset$ and $\rng(\ephi v) \cap \bigcup \rng(\ephi s) = \emptyset$;
	\item[(PGE-5)\label{def:pge-5}]
			$\forall r \in m_G : \forall s \in n_G : \prnt_H^*\circ \ephi r(r) \cap \ephi s(s) = \emptyset$;
	\item[(PGE-6)\label{def:pge-6}]
		$\ephi c \circ \restr{\prnt_G^{-1}}{V_G} = \prnt_H^{-1}\circ \ephi v$;
	\item[(PGE-7)\label{def:pge-7}]
				$\ctrl_G = \ctrl_H \circ \ephi v$;
	\item[(PGE-8)\label{def:pge-8}]
				$\forall c \in n_G \uplus V_G : \forall c' \in \ephi c(c) : 
				(\ephi f \circ \prnt_G)(c) = \prnt_H(c')$;
	\end{description}
	where $\prnt_H^*(c) = \bigcup_{i < \omega} \prnt^i(c)$,
	$\ephi f \defeq \ephi v \uplus \ephi{r}$, and
	$\ephi c \defeq \ephi v \uplus \ephi{s}$.
\end{definition}

Conditions in the above definition follows the structure
of Definition~\ref{def:lge}, the main notable difference is
\ref{def:pge-5} which states that the image of a root
can not be the descendant of the image of another. 
Conditions \ref{def:pge-1}, \ref{def:pge-2} and \ref{def:pge-3} are on the three sub-maps composing the
embedding; conditions \ref{def:pge-4} and 
\ref{def:pge-5} ensure that no components are identified;
\ref{def:pge-6} imposes surjectivity on children and the last two conditions
require the guest structure to be preserved by the embedding map.

\paragraph{Bigraph}
Finally, bigraph embeddings can now be defined as maps being composed 
by an embedding for the link graph with one for the place graph 
consistently with the interplay of these two substructures. In 
particular, the interplay is captured by a single additional 
condition ensuring that points in the image of an inner names 
reside in the parameter defined by the place graph embedding 
(i.e.~are inner names or ports of some node under a site image).

\begin{definition}[Bigraph embedding {\cite[Def~7.5.14]{hoesgaard:thesis}}]\label{def:bge}
	Let $G : \face{n_G,X_G} \to \face{m_G,Y_G}$ and 
	$H : \face{n_H,X_H} \to \face{m_H,Y_H}$ be two concrete
	bigraphs. A \emph{bigraph embedding} $\phi : G \emb H$
	is a map given by a place graph embedding 
	$\ephi P : G^P\emb H^P$ and a link graph embedding
	$\ephi L : G^L\emb H^L$ subject to the consistency
	condition:
	\begin{description}
		\item[(BGE-1)\label{def:bge-1}]
			$\rng(\ephi i) \subseteq X_H \uplus 
			\{(v,i) \in P_H \mid \exists s \in n_G : k \in \mathbb{N}:
			\prnt_H^k(v) \in \ephi s(s)\}$.
	\end{description}
\end{definition}

\paragraph{NP-completeness}
Despite their apparent complexity, the conditions maps have to satisfy to be
considered bigraph embeddings may give some information and 
guidance in the construction of these maps.
However the problem remains NP-complete as demonstrated in \cite{bmr:tgc14}.
We recall their results to make this paper self contained.  
The authors focus on labelled forest embedding which covers the case 
of place graphs embeddings but not link graphs.
In Section~\ref{sec:lge-csp} we prove that the link graph 
embedding problem corresponds to an admissibility
problem for a specific flow network. Therefore, the
result presented in \cite{bmr:tgc14} will suffice to justify 
our approach.

\looseness=-1
To prove that the labelled forest pattern is NP-complete,
in \cite[§3]{bmr:tgc14} a reduction from \textsc{3-Sat} is provided.
The proposed reduction uses the \textsc{RainbowAntichain} problem
as a middle step (introduced in \emph{loc.~cit.}). An instance of this problem
is a tree $\mathcal{T(V,E)}$ with nodes $\mathcal V$ and edges
$\mathcal E$, and a finite set of colours $\mathcal P$, said palette.
Some of the nodes in $\mathcal T$ have been coloured with one or more
colours taken from $\mathcal P$. The problem asks to decide whatever
exists a colourful subset of nodes $\mathcal R \subset \mathcal V$ 
where each colour $c$ of $\mathcal P$ has exactly one representative node
coloured with $c$ and for no pair of $u,v \in \mathcal R$ of distinct
nodes $u$ is an ancestor of $v$.

\begin{theorem}[\hspace{-.2ex}{\cite[Th.~8]{bmr:tgc14}}]
  The \textsc{RainbowAntichain} problem is NP-complete.
\end{theorem}

It is the straightforward to see that an instance $\mathcal T$, 
$\mathcal P = (c_0,\dots,c_{n-1})$ of \textsc{RainbowAntichain}
can be reduced to a forest pattern matching, namely, one that
embeds the forest $(c_0[0],\dots,c_{n-1}[n-1])$ -- every tree has
 only a node, labelled with a colour of the palette, and a hole/site
-- into $\mathcal T$. This states that the forest pattern matching 
problem is NP-complete. Formally,
\begin{theorem}[\hspace{-.2ex}{\cite[Th.~9]{bmr:tgc14}}]
  The labelled forest embedding problem is NP-complete.
\end{theorem}

This proves that deciding the existence of a place graph embedding
(which can be seen as labelled forest pattern matching) of a given
guest into a given host is NP-complete. Moreover, we are interested 
in listing all of them thus making CSP a viable approach.

\section{Implementing the embedding problem in CSP}
\label{sec:csp}

In this Section we present the main contribution of the
paper i.e.~a constraint satisfaction problem that models
bigraph embedding problem. The encoding is based
solely on integer linear constraints and is proven
to be sound and complete.

Initially, we present the encoding for the link graph
embedding problem and for the place graph embedding problem.
Then we combine them providing some additional
``gluing constraints'' to ensure the consistency
of the two sub-problems. The resulting encodings contains
34 constraint families (reflecting the size of the 
problem definition, cf.~in Section~\ref{sec:emb})
and hence taking advantage of the orthogonality of
link and place structures is mandatory for the sake
of both exposition and adequacy proofs. We shall remark
that, despite the constraint families are quite numerous,
the overall number of variables and constraints produced 
by the encoding is polynomially bounded with respect to 
the support cardinality of the involved bigraphs.

\subsection{Link Graphs}\label{sec:lge-csp}

\begin{figure}[t]
	\centering
	\begin{tikzpicture}[>=stealth,scale=.7,
		dot/.style={circle,fill=black,minimum size=5pt,inner sep=0, outer sep=3pt},
		var/.style={->},
		lnk/.style={draw,black!30!green,thick}]
		\node[dot] (n0) at (0,0) {}; %
		\node[dot] (n1) at (-1,-1) {};
		\node[dot] (n2) at (1,1) {};
		\node[dot] (n3) at (4,0) {}; %
		\node[dot] (n4) at (3,1) {};
		\node[dot] (n5) at (5,-1) {};
		\node[dot] (n6) at (0,-4) {}; %
		\node[dot] (n7) at (-1,-3) {};
		\node[dot] (n8) at (1,-5) {};
		\node[dot] (n9) at (4,-4) {}; %
		\node[dot] (n10) at (3,-5) {};
		\node[dot] (n11) at (5,-3) {};
		
		\draw[var,bend right] (n0) to (n3);
		\draw[var,bend left] (n0) to (n6);
		\draw[var,bend left] (n0) to (n7);
		\draw[var,bend left] (n0) to (n8);
		\draw[var,bend left] (n9) to (n3);
		\draw[var,bend left] (n10) to (n3);
		\draw[var,bend left] (n11) to (n3);
				
		\coordinate (l0) at (2,-5.7);
		\coordinate (l1) at (2,-5.8);
		
		\draw [lnk,out=-90,in=180] (n6) to (l0);		
		\draw [lnk, out=-80,in=180] (n8) to (l0);
		\draw [lnk,<-,out=-120,in=0] (n10) to (l0);
		\draw [lnk,out=-90,in=180] (n7) to (l1);
		\draw [lnk,<-,out=-90,in=0] (n11) to (l1);
		
		\node[] (p0) at (2,2) {Network variables};
		\draw[] (p0.south east) -- (p0.south west);
		\draw[] ($(p0.south east)!.5!(p0.south west)$) -- ++(0,-2);
		
		\node[] (p1) at (2,-6.6) {Guest linking};
		\draw[] (p1.north east) -- (p1.north west);
		\draw[] ($(p1.north east)!.5!(p1.north west)$) -- ++(0,.3);
		
		\node[] (p2) at (-3,1) {Host points};
		\draw[shorten >=6pt] (p2.south west) -- (p2.south east) -- ++ (1,0)-- (n0);
		
		\node[] (p3) at (7.5,1) {Host handles};
		\draw[shorten >=6pt] (p3.south east) -- (p3.south west) -- ++ (-1.3,0)-- (n3);

		\node[] (p4) at (-2.85,-5) {Guest points};
		\draw[shorten >=20pt] (p4.south west) -- (p4.south east) -- ++ (.25,0)-- (n6);
		
		\node[] (p5) at (7.35,-5) {Guest handles};
		\draw[shorten >=20pt] (p5.south east) -- (p5.south west) -- ++ (-.6,0)-- (n9);
		
	\end{tikzpicture}
	\caption{Schema of the multi-flux network encoding.}
	\label{fig:lge-flux-net}
\end{figure}

\begin{figure}[!t]
	\begin{align}
	% variabili di rete
	N_{h,h'} \in \{0,\dots,|\link^{-1}_H(h')|\} &\qquad 
		\begin{array}{l}
			h \in E_G \uplus Y_G,\ h' \in E_H \uplus Y_H
		\end{array}
		\label{eq:lge-var-1}\\
	N_{p,h'} \in \{0,1\} &\qquad 
		\begin{array}{l}
			h' \in E_H \uplus Y_H,\ p \in \link_H^{-1}(h')
		\end{array}
		\label{eq:lge-var-2}\\
	N_{p,p'} \in \{0,1\} &\qquad 
		\begin{array}{l}
			p' \in X_G \uplus P_G, \ p \in X_H\uplus P_H
		\end{array}
		\label{eq:lge-var-3}\\
	% variabili di flusso
	F_{h,h'} \in \{0,1\} &\qquad 
		\begin{array}{l}
			h \in E_G \uplus Y_G,\ h' \in E_H \uplus Y_H
		\end{array}
		\label{eq:lge-var-4}\\
	% i punti di H sono sorgenti
	\sum_{k}N_{p,k} = 1 &\qquad 
		\begin{array}{l}
			p \in X_H\uplus P_H
		\end{array}
		\label{eq:lge-cs-1}\\
	% le maniglie di H sono pozzi
	\sum_{k}N_{k,h} = |\link^{-1}_H(h)| &\qquad 
		\begin{array}{l}
			h \in E_H\uplus Y_H
		\end{array}
		\label{eq:lge-cs-2}\\
	% conservazione del flusso nel redex
	\sum_{k}N_{h,k} = \sum_{p \in \link_G^{-1}(h)}\sum_{k}N_{k,p} &\qquad 
		\begin{array}{l}
			h \in E_G\uplus Y_G
		\end{array}
		\label{eq:lge-cs-3}\\
	% le porte di G sono "sorgenti"
	\sum_{k}N_{k,p} = 1 &\qquad 
		\begin{array}{l}
			p \in X_G\uplus P_G
		\end{array}
		\label{eq:lge-cs-4}\\
	N_{p,p'} = 0 & \qquad
		\begin{array}{l}
			p' \in P_G,\ p \in X_H
		\end{array}
		\label{eq:lge-cs-5}\\
	% legame tra rete e flusso
	\frac{N_{h,h'}}{|\link_H^{-1}(h')|} \leq F_{h,h'} \leq N_{h,h'} &\qquad
		\begin{array}{l}
			h \in E_G\uplus Y_G,\ h' \in E_H\uplus Y_H, \\
			|\link_G^{-1}(h)| > 0,\ |\link_H^{-1}(h')| > 0
		\end{array}
		\label{eq:lge-cs-6}\\
	N_{p,p'} \leq F_{h,h'} & \qquad
		\begin{array}{l}
			h \in E_G\uplus Y_G,\ h' \in E_H\uplus Y_H, \\
			p \in \link_G^{-1}(h),\ p' \in \link_H^{-1}(h')
		\end{array}
		\label{eq:lge-cs-7}\\
	F_{h,h'} \leq \sum_{\substack{p \in \link_G^{-1}(h)\\ p' \in \link_H^{-1}(h')}} N_{p,p'} & \qquad
		\begin{array}{l}
			h \in E_G\uplus Y_G,\ h' \in E_H\uplus Y_H, \\
			\link_G^{-1}(h) \cup \link_H^{-1}(h') \neq \emptyset
		\end{array}
		\label{eq:lge-cs-8}\\
	% separazione dei flussi
	\sum_{k}F_{h,k} = 1 &\qquad 
		\begin{array}{l}
			h \in E_G\uplus Y_G
		\end{array}
		\label{eq:lge-cs-9}\\
	% uniformità dei tipi di flusso
	N_{p,h'} + F_{h,h'} \leq 1 & \qquad
		\begin{array}{l}
			h \in E_G,\ h' \in E_H\uplus Y_H,\ 
			p \in \link_H^{-1}(h')
		\end{array}
		\label{eq:lge-cs-10}\\		
	F_{h,h'} + F_{h'',h'} \leq 1 & \qquad
		\begin{array}{l}
			h \in E_G,\ h' \in Y_H,\ h'' \in Y_G 
		\end{array}
		\label{eq:lge-cs-11}\\
	F_{h,h'} = 0 & \qquad
		\begin{array}{l}
			h \in E_G,\ h' \in Y_H 
		\end{array}
		\label{eq:lge-cs-12}\\
	F_{h,h'} \leq 1 & \qquad
		\begin{array}{l}
			h \in E_G \uplus Y_G,\ h' \in E_H 
		\end{array}
		\label{eq:lge-cs-13}\\
	N_{p,p'} = 0 & \qquad
		\begin{array}{l}
			v \in V_G,\ v' \in V_H, \ctrl_G(v) = \ctrl_H(v) = c\, \\
			i\neq i' \leq c,\ p = (v,i),\ p' = (v',i')
		\end{array}	
		\label{eq:lge-cs-14}\\
	N_{p,p'} = 0 & \qquad
		\begin{array}{l}
			v \in V_G,\ v' \in V_H, \ctrl_G(v) \neq \ctrl_H(v),\\
			p = (v,i),\ p' = (v',i')
		\end{array}	
		\label{eq:lge-cs-15}\\
	\sum_{j \leq c} 
		N_{(v,j),(v',j)} = c\cdot N_{p,p'} & \qquad
		\begin{array}{l}
			v \in V_G,\ v' \in V_H, \ctrl_G(v) = \ctrl_H(v) = c,\\
			i \leq c,\ p = (v,i),\ p' = (v',i)
		\end{array}
		\label{eq:lge-cs-16}
	\end{align}
	\caption{Constraints of \textsc{LGE}[$G,H$].}
	\label{fig:lge-csp}
\end{figure}

Let us fix the guest and host concrete bigraphs:
$G : X_G \to Y_G$ and $H : X_H \to Y_H$. We characterize
the embeddings of $G$ into $H$ as the solutions of a suitable
multi-flux problem which we denote as \textsc{LGE}[$G,H$]. 
The main idea is to see the host points (i.e.~ports and inner names)
and handles (i.e.~edges and outer names) as sources and
sinks respectively: each point outputs a flux unit and each handle
inputs one unit for each point it links. Units flows towards each point handle
following $H$ hyper-edges and optionally taking a ``detour'' along the 
linking structure of the guest $G$ (provided that some conditions regarding
structure preservation are met).
Figure~\ref{fig:lge-flux-net} and Figure~\ref{fig:lge-csp}
contain a sketch of the flux problem and its formal definition
respectively.

The flux network reflects the linking structure and contains an edge
connecting each point to its handle; these edges have an integer capacity
limited to $1$ and are represented by the variables defined in \eqref{eq:lge-var-2}.
The remaining edges of the network are organised in two complete biparted graphs:
one between guest and host handles and one between guest and host points.
Edges of the first sub-network are described by the variables in \eqref{eq:lge-var-1}
and their capacity is bounded by the number of points linked by the host handle
since this is the maximum acceptable flux and corresponds to the case where
each point passes through the same hyper-edge of the guest link graph.
Edges of the second sub-network are described by the variables in \eqref{eq:lge-var-3}
and, like the first group of links, have their capacity limited to $1$; 
to be precise, some of these variables will never assume a value different from $0$
because guest points can receive flux from anything but the host ports (as
expressed by constraint \eqref{eq:lge-cs-5}).
Edges for the link structure of the guest are presented implicitly in the
flux preservation constraints (see constraint \eqref{eq:lge-cs-3}). In order to fulfil the injectivity
conditions of link embeddings, some additional \emph{flux variables}
(whereas the previous are \emph{network variables})
are defined by \eqref{eq:lge-var-4}. These are used to
keep track and separate each flux on the bases of the points handle\footnote{%
The problem can be presented without the additional flux variables, but we found
this approach more readable.}.

The constraint families \eqref{eq:lge-cs-1} and \eqref{eq:lge-cs-2}
define the outgoing and ingoing flux of host points and handles respectively.
The firsts have to send exactly one unit considering every edge they
are involved into and the seconds receive one unit for each of their
point regardless if this unit comes from the point directly or from
a handle of the guest. 

The linking structure of the guest graph is
encoded by the constraint family \eqref{eq:lge-cs-3} which states
that flux is preserved while passing through the guest i.e.~the output
of each handle has to match the overall input of the points it connects.

\looseness=-1
Constraints \eqref{eq:lge-cs-4}, \eqref{eq:lge-cs-5}, 
\eqref{eq:lge-cs-14}, \eqref{eq:lge-cs-15} and  \eqref{eq:lge-cs-16} 
shape the flux in the sub-network linking guest and host points.
Specifically, \eqref{eq:lge-cs-4} requires that each point from the guest 
receives exactly one unit or, the other way round, that guest points are
assigned with disjoint sets of points from the host.
Constraints \eqref{eq:lge-cs-5}, \eqref{eq:lge-cs-14} and \eqref{eq:lge-cs-15}
disable edges between guest ports and host inner names, between mismatching 
ports of matching nodes and between ports of mismatching nodes.
Finally, the flux of ports belonging to the same node has to act compactly
i.e.~if there is flux between the $i$-th ports of two nodes, then, there
should be flux between every other matching ports as expressed by
\eqref{eq:lge-cs-16}.

\looseness=-1
Constraints \eqref{eq:lge-cs-6}, \eqref{eq:lge-cs-7} and \eqref{eq:lge-cs-8}
relates flux and network variables ensuring that the formers assume
a true value if, and only, if there is actual flux between the corresponding
guest and host handles. In particular, \eqref{eq:lge-cs-7}  propagates
the information about the absence of flux between handles disabling the sub-network
linking handles points and, \emph{vice versa},  \eqref{eq:lge-cs-8}
propagates the information in the other way disabling flux between handles
if there is no flux between their points. 

The remaining constraints prevent fluxes from mixing. Constraint \eqref{eq:lge-cs-9}
requires guest handles to send their output to exactly one destination
thus renders the sub-network between handles a function assigning guest
handles to host handles. This mapping is subject to some additional
conditions when edges are involved: \eqref{eq:lge-cs-12} and \eqref{eq:lge-cs-13} 
ensure that the edges are injectively mapped to edges only, \eqref{eq:lge-cs-11}
forbids host outer names to receive flux from an edge and an outer name at the same time.
Finally, constraint \eqref{eq:lge-cs-10} states that the output of host points
cannot bypass the guest if there is flux between its handle and an edge
from the guest.

\paragraph{Adequacy}
Let $\vec N$ be a solution of \textsc{LGE[$G,H$]}. The corresponding
link graph embedding $\phi : G \emb H$ is defined as follows:
\begin{align*}
	\ephi{v}(v)&\defeq 
		v'  \in V_H : \exists i : N_{(v,i),(v',i)} = 1 \\
	\ephi{e}(e)&\defeq 
		e' \in E_H : F_{e,e'} = 1 \\
	\ephi{o}(y)&\defeq 
		y' \in Y_H : F_{y,y'} = 1 \\
	\ephi{i}(x)&\defeq 
		\{x' \in X_H \uplus P_H \mid N_{x',x} = 1\}
\end{align*}
The components of $\phi$ just defined are well-given and
compliant with Definition~\ref{def:lge}.
On the other way round, let ${\phi : G \emb H}$ be a link
graph embedding. The corresponding solution $\vec N$ 
of \textsc{LGE[$G,H$]} is defined as follows:
\begin{align*}
	N_{p,p'}&\defeq
	\begin{cases}
		1 & \mbox{if } p' \in X_G \land p \in \ephi{i}(p') \\
		1 & \mbox{if } p' = (v,i) \land p = (\ephi{v}(v),i) \\
		0 & \mbox{otherwise}
	\end{cases}\\
	N_{p,h'}&\defeq
	\begin{cases}
		1 & \mbox{if } h' = \link_H(p) \land \nexists p' : N_{p,p'} = 1 \\
		0 & \mbox{otherwise}
	\end{cases}\\
	N_{h,h'}&\defeq
	\begin{cases}
		1 & \mbox{if } h \in E_H \land h' \in E_G \land h = \ephi{e}(h') \\
		1 & \mbox{if } h \in Y_H \land h' \land h = \ephi{o}(h') \\
		0 & \mbox{otherwise}
	\end{cases}
\end{align*}
Clearly $F_{h,h'} = 1 \defiff N_{h,h'} \neq 0$. Then it is
easy to check that every constraint of \textsc{LGE[$G,H$]}
is satisfied by the solution just defined.

The constraint satisfaction problem in Figure~\ref{fig:lge-csp}
is sound and complete with respect to the link graph embedding problem
given in Definition~\ref{def:lge}.
\begin{proposition}[Adequacy of \textsc{LGE}]
\label{prop:lge-adequacy}
For any two concrete link graphs $G$ and $H$,
there is a bijective correspondence between
the link graph embeddings of $G$ into $H$ and
the solutions of \textsc{LGE[$G,H$]}.
\end{proposition}

\subsection{Place Graphs}

\begin{figure}[t]
	\begin{align}
		M_{h,g} \in \{0,1\} & \qquad 
			\begin{array}{l}
				g \in n_G \uplus V_G \uplus m_G,\ \\
				h \in n_H \uplus V_H \uplus m_H
			\end{array}
			\label{eq:pge-var-1}\\
		M_{h,g} = 0 & \qquad
			\begin{array}{l}
				g \in n_G \uplus V_G,\ h \in m_H
			\end{array}
			\label{eq:pge-cs-1}\\
		M_{h,g} = 0 & \qquad
			\begin{array}{l}
				g \in V_G \uplus m_G,\ h \in n_H
			\end{array}
			\label{eq:pge-cs-2}\\
		M_{h,g} = 0 & \qquad
			\begin{array}{l}
				g \in V_G ,\ h \in V_H,\\
				\ctrl_G(g) \neq \ctrl_H(h)
			\end{array}
			\label{eq:pge-cs-3}\\
		M_{h,g} = 0 & \qquad
			\begin{array}{l}
				g \in m_G ,\ h \notin m_H,\\ 
				v \in \prnt_H^*(h) \cap V_G,\\
				\ctrl_G(v) \notin \Sigma_a
			\end{array}
			\label{eq:pge-cs-4}\\
		M_{h,g} \leq M_{h',g'} & \qquad
			\begin{array}{l}
				g \notin m_G,\ g' \in \prnt_G(g),\\ 
				h \notin m_H,\ h' \in \prnt_H(h)
			\end{array}
			\label{eq:pge-cs-5}\\
		\sum_{h \in V_H \uplus m_H} M_{h,g} = 1 &\qquad 
			\begin{array}{l}
				g \in m_G
			\end{array}
			\label{eq:pge-cs-6}\\
		\sum_{h \in n_H \uplus V_H} M_{h,g} = 1 &\qquad 
			\begin{array}{l}
				g \in V_G
			\end{array}
			\label{eq:pge-cs-7}\\
		m_G \cdot \sum_{g \in n_G\uplus V_G} M_{h,g} + \sum_{g \in m_G} M_{h,g}
		\leq m_G &\qquad
			\begin{array}{l}
				h \in V_H
			\end{array}
			\label{eq:pge-cs-8}\\
		|\prnt_H^{-1}(h)|\cdot M_{h,g} \leq 
			\sum_{\substack{h' \in \prnt_H^{-1}(h),\\ g' \in \prnt_G^{-1}(g)}} 
			M_{h'\!,g'} & \qquad
			\begin{array}{l}
				g \in V_G,\ h \in V_H
			\end{array}
			\label{eq:pge-cs-9}\\
		|\prnt_G^{-1}(g)\setminus n_G|\cdot M_{h,g} \leq 
			\sum_{\substack{h' \in \prnt_H^{-1}(h)\setminus n_h,\\ g' \in \prnt_G^{-1}(g) \setminus n_g}} 
			M_{h'\!,g'} & \qquad
			\begin{array}{l}
				g \in m_G,\ h \in V_H
			\end{array}
			\label{eq:pge-cs-10}\\
		M_{h,g} + \sum_{\substack{h' \in \prnt_H^*(h),\\g'\in m_G}}M_{h',g'} \leq 1 &\qquad
			\begin{array}{l}
				g \in V_G,\ h\in V_H
			\end{array}
			\label{eq:pge-cs-11}
	\end{align}
	\caption{Constraints of \textsc{PGE}[$G,H$].}
	\label{fig:pge-csp}
\end{figure}

Let us fix the guest and host place graphs:
$G : n_G \to m_G$ and $H : n_H \to m_H$. We characterize
the embeddings of $G$ into $H$ as the solutions of the
constraint satisfaction problem in Figure~\ref{fig:pge-csp}.
The problem is a direct encoding of Definition~\ref{def:pge}
as a matching problem presented, as usual, as a biparted graph.
Sites, nodes and roots of the two place graphs are represented as nodes
and parted into the guest and the host ones. For convenience of exposition, 
graph is complete.

Edges are modelled by the boolean variables defined in \eqref{eq:pge-var-1};
these are the only variables used by the problem. So far a solution is
nothing more than a relation between the components of guest and host
containing only those pairs connected by an edge assigned a non-zero value.
To capture exactly those assignments that are actual place graph embeddings
some conditions have to be imposed.

Constraints \eqref{eq:pge-cs-1} and \eqref{eq:pge-cs-2}
prevent roots and sites from the host to be matched with nodes or sites
and nodes or roots respectively. \eqref{eq:pge-cs-3} 
disables matching between nodes decorated with different controls.
Constraint \eqref{eq:pge-cs-4} prevents any matching for host nodes
under a passive context (i.e.~have an ancestor labelled with a passive control).
\eqref{eq:pge-cs-5} propagates the matching along the parent map from children
to parents. Constraints \eqref{eq:pge-cs-6} and \eqref{eq:pge-cs-7}
ensure that the matching is a function when restricted to guest nodes and roots
(the codomain restriction follows by \eqref{eq:pge-cs-1} and \eqref{eq:pge-cs-2}).
\eqref{eq:pge-cs-8} says that if a node from the host cannot be
matched with a root or a node/site from the guest at the same time;
moreover, if the host node is matched with a node then it cannot be matched
to anything else.

The remaining constraints are the counterpart of \eqref{eq:pge-cs-5}
and propagate matchings from parents to children.
\eqref{eq:pge-cs-9} applies on matchings between nodes and says
that if parents are matched, then children from the host node are covered
by children from the guest node. In particular, the matching is a perfect assignment
when restricted to guest children that are nodes  (because of \eqref{eq:pge-cs-8})
and is a surjection on those that are sites.
\eqref{eq:pge-cs-10} imposes a similar condition on matchings between
guest roots and host nodes. Specifically, it says that the matching have
to cover child nodes from the guest (moreover, it is injective on them)
leaving child sites to match whatever remains ranging from nothing to all
unmatched children. Finally, \eqref{eq:pge-cs-11} prevent matching from happening inside a parameter.

\paragraph{Adequacy}
Let $\vec M$ be a solution of \textsc{PGE[$G,H$]}. The corresponding
place graph embedding $\phi : G \emb H$ is defined as follows:
\begin{align*}
	\ephi{v}(g)&\defeq 
		h  \in V_H : \exists i : M_{h,g} = 1 \\
	\ephi{s}(g)&\defeq 
		\{h \in n_h \uplus V_H \mid M_{h,g} = 1\}\\
	\ephi{r}(g)&\defeq 
		h \in m_H \uplus V_H : M_{h,g} = 1
\end{align*}
The components of $\phi$ just defined are well-given and
compliant with Definition~\ref{def:pge}.

On the opposite direction, let ${\phi : G \emb H}$ be a place
graph embedding. The corresponding solution $\vec M$ 
of \textsc{PGE[$G,H$]} is defined as follows:
\[
	M_{h,g} \defeq
	\begin{cases}
		1 & \mbox{if } g \in V_G \land h = \ephi{v}(g) \\
		1 & \mbox{if } g \in m_G \land h = \ephi{r}(g) \\
		1 & \mbox{if } g \in n_G \land h \in \ephi{s}(g) \\
		0 & \mbox{otherwise}
	\end{cases}
\]
It is easy to check that every constraint of \textsc{PGE[$G,H$]}
is satisfied by the solution just defined.

The constraint satisfaction problem in Figure~\ref{fig:pge-csp}
is sound and complete with respect to the place graph embedding problem
given in Definition~\ref{def:pge}.
\begin{proposition}[Adequacy of \textsc{PGE}]
\label{prop:pge-adequacy}
For any two concrete place graphs $G$ and $H$,
there is a bijective correspondence between
the place graph embeddings of $G$ into $H$ and
the solutions of \textsc{PGE[$G,H$]}.
\end{proposition}

\subsection{Bigraphs}
Let $G : \face{n_G,X_G} \to \face{m_G,Y_G}$ and $H : \face{n_H,X_H} \to \face{m_H,Y_H}$ 
be two concrete bigraphs. By taking advantage of the orthogonality
of the link and place structure we can define the constraint satisfaction
problem capturing bigraph embeddings by simply composing the constraints
given above for the link and place graph embeddings and by adding just two 
consistency constraints to relate the solutions of the two problems.

These additional constraint families are reported in Figure~\ref{fig:bge-csp}.
The family \eqref{eq:bge-cs-1} ensures that solutions for 
\textsc{LGE[$G,H$]} and \textsc{PGE[$G,H$]} agree on nodes since 
the map $\ephi{v}$ has to be shared by the corresponding 
link and place embeddings. 
The family \eqref{eq:bge-cs-2} ensures that ports are in the image
of inner names (i.e.~send their flux unit to them) only if their node is part of
the parameter i.e.~only if it is matched to a site from the guest or it descends from a node that is so. 

\begin{figure}[t]
	\begin{align}
	M_{v,v'} = N_{p,p'} &\qquad
		v \in V_H,\ v' \in V_G,\ 
		p = (v,k),\ p' = (v',k)
		\label{eq:bge-cs-1}\\
	\sum_{p' \in X_G} N_{p,p'} \leq \sum_{\substack{h \in \prnt_H^*(v),\\g \in n_G}} M_{h,g} &\qquad
		v \in V_H,\ p = (v,k)
		\label{eq:bge-cs-2}
	\end{align}
	\caption{Constraints of \textsc{BGE}[$G,H$].}
	\label{fig:bge-csp}
\end{figure}

It is easy to check that \eqref{eq:bge-cs-2} corresponds exactly to
condition \ref{def:bge-1}. Therefore, from Proposition~\ref{prop:lge-adequacy} and
Proposition~\ref{prop:pge-adequacy}, the constraint satisfaction problem defined by
Figures~\ref{fig:lge-csp}, Figure~\ref{fig:pge-csp} and Figure~\ref{fig:bge-csp}
is sound and complete with respect to the bigraph embedding problem
given in Definition~\ref{def:bge} as stated by below.
\begin{proposition}[Adequacy of \textsc{BGE}]
\label{prop:bge-adequacy}
For any two concrete bigraphs $G$ and $H$,
there is a bijective correspondence between
the bigraph embeddings of $G$ into $H$ and
the solutions of \textsc{BGE[$G,H$]}.
\end{proposition}

\section{Conclusions and future works}\label{sec:concl}

In this paper, we have presented a sound and complete algorithm for solving the
bigraph embedding problem, based on a constraint satisfaction problem.
The resulting model is compact and composed by a number of variables
and linear constraints, polynomially bounded by the size of the guest
and host bigraphs. Remarkably, this algorithm does not require hosts
to be ground.

The CSP approach offers a great flexibility, e.g.~allowing to move
execution strategies upstream into CSP calculation. An example can
be found in this paper: in practice, the notion of active/passive
contexts defines an execution strategy.

This approach naturally suggests several interesting extensions of the
bigraphic model itself. We can think of \emph{weighted} bigraphs,
where nodes and edges can be given ``weights'', and firing a rule has
a cost related to the nodes involved.  This extension would lead us
to consider
``approximated embeddings'': a guest can be embedded up-to some distance
based on costs (e.g. missing controls, or quantitative differences
between controls,\dots). In both cases, our implementation can be
straightforwardly adapted just by adding a cost function, in order to
give the optimal matchings; this would be far more efficient and
easier to implement than dealing with these issues at the level of
strategies (if possible at all).

The proposed approach can be easily applied also to extensions of
the bigraphs, e.g.~directed bigraphs \cite{gm:mfps07}, bigraphs with
sharing \cite{cs:ifm12} or local bigraphs.  An interesting direction
would be to extend the algorithm also to stochastic and probabilistic
bigraphs \cite{kmt:mfps08} which will offer useful modelling and
verification tools for quantitative aspects, e.g.~for biological
problems \cite{bgm:biobig,dhk:fcm}. 

The algorithm has been successfully integrated into LibBig,
an extensible library for manipulating bigraphical reactive systems.
This library can be used for implementing a wide range of tools and
it can be adapted to support several extensions of bigraphs.
Approximated and weighted embeddings are supported too but are still
an experimental feature. In fact, the theoretical foundations and implications
of these extensions have not been fully investigated yet suggesting
another line of research.

The empirical evaluation of the implementation available in LibBig
looks promising. It cannot be considered a rigorous experimental
validation mainly because performance depends on the solver, and the
model is not optimized for any particular solver. Moreover, up to now
there are no ``official'' (or ``widely recognized'') benchmarks for
the bigraph embedding problem.
An algorithm for the bigraph matching problem is proposed in
\cite{gdbm13:indmatch}, based on a translation into a term matching
problem; this algorithm is at the core of BPLTool.  In
\cite{sevegnani2010sat} Sevegnani \emph{et.~al.} presented a SAT based
algorithm for deciding the matching problem in bigraphs with sharing.  To the
best of the authors knowledge, the approach of \cite{sevegnani2010sat}
is the nearest to the solution presented in this paper. However, an
accurate and fair comparison of algorithms for computing bigraph
embeddings/matchings is difficult because of the aforementioned
reasons and because of the different bigraph variants these algorithms
deal with (e.g., the matching algorithm of the model checker BigMC
does not support inner names).

A decentralized algorithm for computing bigraphical embeddings
has been proposed in \cite{mpm:gcm14w,mpm:gcm14j} and is at the core of
the \emph{distributed bigraphical machine}.

\paragraph{Acknowledgements} 
We thank Alessio Mansutti and the participants to MeMo'14 for fruitful
discussions on preliminary version of this paper.  
This work is partially supported by MIUR PRIN project 2010LHT4KM, \emph{CINA}.

\FloatBarrier

{

}

\clearpage
\appendix

\section{NP-completeness of the bigraph embedding problem}\label{sec:npcomp}
In \cite{bmr:tgc14} the authors proved that the labelled forest embedding problem
is NP-Complete. This result covers the case of place graphs embeddings but not link graphs. However, the latter correspond to an admissibility problem
for a specific flow network and hence their results will suffice to justify our approach.

To prove that the labelled forest pattern is NP-complete,
in \cite[§3]{bmr:tgc14} a reduction from \textsc{3-Sat} is provided.
The proposed reduction uses the \textsc{RainbowAntichain} problem
as a middle step. An instance of this problem (introduced in \cite{bmr:tgc14})
is a tree $\mathcal{T(V,E)}$ with nodes $\mathcal V$ and edges
$\mathcal E$, and a finite set of colours $\mathcal P$, said palette.
Some of the nodes in $\mathcal T$ have been coloured with one or more
colours taken from $\mathcal P$. The problem asks to decide whatever
exists a colourful subset of nodes $\mathcal R \subset \mathcal V$ 
where each colour $c$ of $\mathcal P$ has exactly one representative node
coloured with $c$ and for no pair of $u,v \in \mathcal R$ of distinct
nodes $u$ is an ancestor of $v$.

\begin{theorem}[\hspace{-.2ex}{\cite[Th.~8]{bmr:tgc14}}]
\label{th:ra-np}
  The \textsc{RainbowAntichain} problem is NP-complete.
\end{theorem}
\begin{proof}
\textsc{RainbowAntichain} is in NP, since, given a set of 
nodes $\mathcal R$, checking whatever $\mathcal R$ us a 
rainbow anti-chain for $\mathcal T$ can be easily done in 
polynomial time by breadth-first visit of $\mathcal T$, 
and for each $v\in \mathcal R$ found, firsti increase the
node counter $nc$, then the colour counter $p[i]$ (with
$1\leq i\leq |\mathcal P|$) if $v$ has a colour $c_i \in 
\mathcal P$. The check fails whether $nc > |\mathcal P|$
or $p[j] = 0$ for some $j$, otherwise $\mathcal R$ is a
rainbow anti-chain for $\mathcal T$.

Let $C = \{c_1,\dots c_m\}$ be an instance of \textsc{3-Sat}
on variables $\{x_1,\dots,x_n\}$. From $C$ we define a 
coloured tree $\mathcal T$ as follows. Let $r$ be the root
node which is left uncoloured. For each variable $x_i$ let
$x_i$ and $\bar x_i$ be child nodes of $r$, and color them 
with fresh colour $c_{x_i}$, distinct for each variable.
For each clause $c_j\in C$. let $c_j^1$, $c_j^2$ and $c_j^3$
be children nodes of $l_i$ in $\mathcal T$ if $c_j$ contains
$l_i$ as negated, and assign to each of them a fresh colour
$c_{c_j}$, distinct for each clause. An example of 
construction for $c_1 = (\bar x_1\lor x_2\lor\bar x_3), 
c_2 = (x_1\lor x_2\lor x_3)$ is shown below.
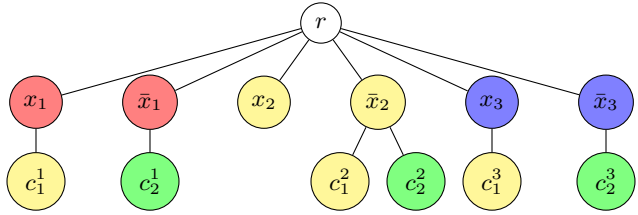
\[
  \begin{tikzpicture}[font=\small,yscale=.7]
    \tikzstyle{every node}=[circle,draw]
    \node {\(r\)}
      child {node[fill=red!50]      {\(x_1\)}
        child {node[fill=yellow!50] {\(c_1^1\)}}}
      child {node[fill=red!50]      {\(\bar x_1\)}
        child {node[fill=green!50]  {\(c_2^1\)}}}
      child {node[fill=yellow!50]   {\(x_2\)}}
      child {node[fill=yellow!50]   {\(\bar x_2\)}
        [sibling distance=1cm]
        child {node[fill=yellow!50] {\(c_1^2\)}}
        child {node[fill=green!50]  {\(c_2^2\)}}}
      child {node[fill=blue!50]     {\(x_3\)}
        child {node[fill=yellow!50] {\(c_1^3\)}}}
      child {node[fill=blue!50]     {\(\bar x_3\)}
        child {node[fill=green!50]  {\(c_2^3\)}}}
    ;
  \end{tikzpicture}
\]
Let $\varphi$ be a truth assignment satisfying the formula
$C$. By construction, selecting only literal nodes $l_i$
which are satisfied by $\varphi$, we obtain a rainbow 
anti-chain $\mathcal R'$ in $\mathcal T$ for the palette
$\{c_{x_i} \mid 1 \leq i \leq n\}$. Now, we extend 
$\mathcal R'$ to $\mathcal R$ adding all clause nodes 
which are not children of an element in $\mathcal R'$.
Such $\mathcal R$ is clearly an anti-chain for $\mathcal
T$, but we must ensure that is colourful and no more than 
one representative per colour is taken. To do this, it
suffices to prove that $\mathcal R$ is colourful, indeed
if a colour occurs more than one in $\mathcal R$ we remove
the others. By hypothesis, each clause $c_J$ is satisfied 
by $\varphi$, hence$c_J$ has at least one literal $l_i$
such that $\varphi(l_i)$ is true. By construction of
$\mathcal T$, there exist a node $c_j^k$, with $1\leq k
\leq 3$, child of $\bar l_i$, hence already in $\mathcal R$.
This holds for all clauses $c_j$, hence $\mathcal R$ is 
colourful.

Controversy, let $\mathcal R$ be a rainbow anti-chain for 
$\mathcal T$. Let the boolean function $\varphi$ over
$\{x_1,\dots,x_n\}$ be defined by $\varphi(x_i) = \mathbb T$
if $x_i$ is a node in $\mathcal R$ and $\varphi(x_i) = 
\mathbb F$ otherwise. Since $\mathcal  R$ has exactly one
representative per colour, no opposite literals are in 
$\mathcal  R$, hence $\varphi$ is a truth assignment for
$C$. By colourfulness of $\mathcal R$, for all colours $c_{c_J}$
$(1\leq j\leq m)$ there exists a node $c_j^k \in \mathcal R$
$(1\leq k\leq 3)$ such that $c_j^k$ has a colour $c_{c_j}$.
By construction on $\mathcal T$, each $c_j^K\in\mathcal R$
is a children of a literal node $l_i\notin \mathcal R$, and
moreover the clause $c_j$ contains $\bar l_j$. Since $l_j\notin
\mathcal R$, by definition $\varphi(\bar l_i) = \mathbb T$,
hence $\varphi(c_j) = \mathbb T$. This holds for all $1\leq j
\leq m$, hence $\varphi$ satisfies $C$.
\end{proof}

In \cite{bmr:tgc14} labelled trees are described by terms of a
language inspired to the ambient calculus and quotiented by 
the usual structural congruence. Then ambients corresponds
to labelled subtrees; the null process to the empty tree;
variables are leaves; parallel processes to siblings 
(with the additional requirement for disjoint variables
to ensure a tree structure). This allows to represent
grafting as (simultaneous) substitution.
An embedding of a labelled forest (with variables $\vv Z$) 
$\vv S(\vv Z)$ into a tree $T$ (denoted as $T \succeq \vv S$) 
can be described by a tree $C(\vv X)$, said context, and a 
forest $\vv D$, whose trees are called parameters, such that 
$T\equiv (C\{\vv S\slash \vv X\})\{\vv D\slash \vv Z\}$. 

It is the straightforward to see that an instance $\mathcal T$, 
$\mathcal P = (c_0,\dots,c_{n-1})$ of \textsc{RainbowAntichain}
can be reduced to a forest pattern matching, namely, one that
embeds the forest $(c_0[0],\dots,c_{n-1}[n-1])$ -- every tree has
 only a node, labelled with a colour of the palette, and a hole/site
-- into $\mathcal T$. This states that the forest pattern matching 
problem is NP-complete. Formally,
\begin{theorem}[\hspace{-.2ex}{\cite[Th.~9]{bmr:tgc14}}]
\label{th:fpm-np}
  The labelled forest embedding problem is NP-complete.
\end{theorem}
\begin{proof}
Given a solution $(C, \vv D)$ for $T \succeq \vv S$, 
checking that $T \equiv (C\{\vv S\slash\vv X\})
\{\vv D\slash\vv Z\}$ corresponds to a tree isomorphism
test, which is in P for \cite{ht:71,hw:74}.
Let a coloured tree $\mathcal T$ and a palette $\mathcal P 
= \{c_1,\dots, c_n\}$ be and instance of the \textsc{
RainbowAntichain} problem. Let us transform $\mathcal T$
into a tree term $T$ as follows. If $T$ is a single node 
$v$ (a leaf) $T$ is the empty tree labelled with $m$ ($m[
\mathsf nil]$) where $m = c$ if $v$ has color $c$, 
otherwise $m = *$, a fresh name not in $\mathcal P$ 
denoting an uncoloured node. If $\mathcal T$ has root $r$ 
and $\mathcal T_1,\dots,\mathcal T_k$ are the (children) 
subtrees of $r$, $T$ is $m[T_1|\dots | T_k ]$, where $m$ 
is as above for $r$, and $T_1,\dots ,T_k$ are transformed
trees of $\mathcal T_1,\dots,\mathcal T_k$.
Suppose $(C, \vv D)$ be a solution for $T \succeq \vv S = 
(c_1[x_1],\dots,c_k[x_n])$. In $C$, each $c_i[x_i]$ is 
grafted into a variable $z_i \in vars(C)$. Since variables
can appear in terms only as leaves, in the transformation 
$\mathcal T$ of $T$, we have found a rainbow anti-chain for 
$\mathcal P$, since the matching forest $\vv S$ has all 
the colors in $\mathcal P$ exactly once. 

Assume that $\mathcal T$ has a rainbow anti-chain $\mathcal R$.
In order to recover context $C$ and parameters $\vv D$, 
which are a solution for $T\succeq (c_1[x_1],\dots,c_k[x_n])$, 
it suffices to apply the construction explained above with 
some adjustments: we obtain $C$ applying the transformation 
from the root of $T$, but if a node in $\mathcal R$ is reached
it is transformed by a fresh variable $z_i$ $(1 \leq i \leq n)$
one for each element in $\mathcal R$; $D_j$’s are recovered 
applying the original transformation starting from the 
subtrees rooted at the children of nodes in $\mathcal R$. 
It is straightforward to prove that $T\equiv(C\{c_1[x_1]\slash 
z_1,\dots, c_k[x_n]\slash z_n\})\{\vv D\slash\vv X\}$, for 
$X = \{x_1,\dots,x_n\}$.
\end{proof}

This proves that deciding the existence of a place graph embedding
(which can be seen as labelled forest pattern matching) of a given
redex into an agent is NP-complete. Moreover, we are interested 
in listing all of them thus making CSP a viable approach.

\end{document}